\documentclass{article}
\usepackage{graphicx}
\usepackage{epstopdf}
\usepackage{dsfont}
\usepackage{stmaryrd}
\usepackage{subfigure}
\usepackage{amsmath}
\usepackage[standard]{ntheorem}
\begin{document}

\title{Proving chaotic behaviour of CBC mode of operation}

\author{Abdessalem Abidi$^1$, Qianxue Wang$^3$, Belgacem Bouallegue$^1$,\\ Mohsen Machhout$^1$, and Christophe Guyeux$^2$
}

\maketitle

\begin{abstract}
The cipher block chaining (CBC) block cipher mode of operation was invented by IBM (International Business Machine) in 1976. It presents a very popular way of encrypting which is used in various applications. In this paper, we have mathematically proven that, under some conditions, the CBC mode of operation can admit a chaotic behaviour according to Devaney. Some cases will be properly studied in order to put in evidence this idea.
\end{abstract}

%\begin{multicols}{2}
\section{Introduction}
\label{intro}

Block ciphers have a very simple principle. They do not treat the original text bit by bit but they manipulate blocks of text -- for example, a block of 64 bits for the DES (Data Encryption Standard) or a block of 128 bits for the AES (Advanced Encryption Standard) algorithm. In fact, the original text is broken into blocks of $\mathsf{N}$ bits. For each block, the encryption algorithm is applied to obtain an encrypted block which has the same size. Then we gather all blocks, which are encrypted separately, to obtain the complete encrypted message. For decryption, we precede in the same way but this time starting from the cipher text to obtain the original message using the decryption algorithm instead of the encryption function. So, it is not sufficient to put anyhow a block cipher algorithm in a program. We can instead use these algorithms in various ways according to their specific needs. These ways are called the block cipher modes of operation. There are several modes of operation and each mode has owns characteristics and its specific security properties. In this paper, we will consider only one of these modes which is the cipher block chaining (CBC) one, and we will study it according to chaos.

%Chaotic iterations have been introduced on the one hand by Chazan and
%Miranker~\cite{chazan1969chaotic} in a numerical analysis context and on the other hand by Robert ~\cite{robert1986discrete} in the discrete dynamical systems framework. 
The chaos theory we consider in this article is the Devaney's topological one~\cite{devaney1989introduction}.In addition to being recognized as one of the
best mathematical definition of chaos, this theory offers a
framework with qualitative and quantitative tools to evaluate
the notion of unpredictability~\cite{bahi2011efficient}. As an application of our fundamental results, we are interested in the area of information safety and security. %We propose in this paper a new approach of security which is based on unpredictability as it is defined by Devaney’s chaos.
The goal of this paper is to study the conditions under which the CBC mode of operation can admit a Devaney's chaotic behavior. 

The remainder of this paper is organized as follows. In Section~\ref{section:BASIC RECALLS}, we will recall some basic definitions concerning chaos and  cipher-block chaining mode of
operation. 
Section~\ref{sec:proof} is devoted to the proofs of chaotic behaviour for the CBC mode. Some cases will be studied in Section~\ref{sec:caseStudy} while results are discussed in the next section. This research work ends by a conclusion section, in which contributions are recalled and some intended future work are proposed.

\section{Basic recalls}
\label{section:BASIC RECALLS}
This section is devoted to basic definitions and terminologies in the field of topological chaos and in the one of block cipher mode of operation.

\subsection{Devaney's Chaotic Dynamical Systems}
\label{subsec:Devaney}
In the remainder of this article, $S^{n}$ denotes the $n^{th}$ term of a sequence $S$ while $\mathcal{X}^\mathds{N}$ is the set of all sequences whose elements belong to $\mathcal{X}$. $V_{i}$
stands for the $i^{th}$ component of a vector $V$. $f^{k}=f\circ ...\circ f$
is for the $k^{th}$ composition of a function $f$. $\mathds{N}$ is the set of natural (non-negative) numbers, while $\mathds{N}^*$ stands for the positive integers $1, 2, 3, \hdots$ Finally, the following
notation is used: $\llbracket1;N\rrbracket=\{1,2,\hdots,N\}$.

Consider a topological space $(\mathcal{X},\tau)$ and a continuous function $f :
\mathcal{X} \rightarrow \mathcal{X}$ on $(\mathcal{X},\tau)$.

\begin{definition}
The function $f$ is \emph{topologically transitive} if, for any pair of open sets
$U,V \subset \mathcal{X}$, there exists an integer $k>0$ such that $f^k(U) \cap V \neq
\varnothing$.
\end{definition}

\begin{definition}
An element $x$ is a \emph{periodic point} for $f$ of period $n\in \mathds{N}$, $n>1$,
if $f^{n}(x)=x$ and $f^k(x) \neq x, 1\le k\le n$. % \linebreak % The set of periodic points of $f$  is denoted $Per(f).$ 
%\end{definition}
%
%\begin{definition}
$f$ is  \emph{regular} on $(\mathcal{X}, \tau)$ if the set of periodic
points for $f$ is dense in $\mathcal{X}$: for any point $x$ in $\mathcal{X}$,
any neighborhood of $x$ contains at least one periodic point.
\end{definition}

\begin{definition}
\label{sensitivity} The function $f$ has \emph{sensitive dependence on initial conditions}
if there exists $\delta >0$ such that, for any $x\in \mathcal{X}$ and any
neighborhood $V$ of $x$, there exist $y\in V$ and $n > 0$ such that
$$d\left(f^{n}(x), f^{n}(y)\right) >\delta .$$
$\delta$ is called the \emph{constant of sensitivity} of $f$.
\end{definition}

\begin{definition}[Devaney's formulation of chaos~\cite{Devaney}]
The function $f$ is  \emph{chaotic} on a metric space $(\mathcal{X},d)$ if $f$ is regular,
topologically transitive, and has sensitive dependence on initial conditions.
\end{definition}

Banks \emph{et al.} have proven in~\cite{Banks92} that when $f$ is regular and transitive on a metric space $(\mathcal{X}, d)$, then $f$ has the property of sensitive dependence on initial conditions. This is why chaos can be formulated too in a topological space $(\mathcal{X}, \tau)$: in that situation, chaos is obtained when $f$ is regular and
topologically transitive.
Note that the transitivity property is often obtained as a consequence of the strong transitivity one, which is defined below.

\begin{definition}
\label{def:strongTrans}
$f$ is \emph{strongly transitive} on $(\mathcal{X},d)$ if, for all point $x,y \in \mathcal{X}$ and for all neighborhood $\mathcal{V}$ of $x$, it exists $n \in \mathds{N}$ and $x'\in \mathcal{V}$ such that $f^n(x')=y$. 
\end{definition}

 \subsection{CBC properties}
 \label{sec:CBC properties}

 Like some other modes of operation, the CBC mode requires not only a plaintext but also an initialization vector (IV) as input.
 In what follows, we will show how this mode of operation works in practice.

 \subsubsection{Initialisation vector IV}
 As what have been already announced, in addition to the plaintext the CBC mode of operation requires an initialization vector in order to randomize the encryption. This vector is used to produce distinct ciphertexts even if the same plaintext is encrypted multiple times, without the need of a slower re-keying process~\cite{huang2013novel}.

An initialization vector must be generated for each execution of the encryption operation, and the same vector is necessary for the corresponding execution of the decryption operation, see Figure~\ref{fig:CBC}. Therefore the IV, or information that is sufficient to calculate it, must be available to each party of any communication. 
The initialization vector does not need to be secret, so the IV, or information sufficient to determine the IV, may be transmitted with the cipher text. In addition, the initialization vector must be unpredictable: for any given plaintext, it must not be possible to predict the IV that will be associated to the plaintext, in advance to the vector generation~\cite{dworkin2001recommendation}.

There are two recommended methods for generating unpredictable IVs. The first method is to apply the forward cipher function, under the same key that is used for the encryption of the plaintext, to a nonce. The nonce must be a data block that is unique to each execution of the
encryption operation. For example, the nonce may be a counter or a message number. The second method is to generate a random data block using a FIPS (Federal Information Processing Standard)-approved random number generator~\cite{dworkin2001recommendation,annex2005approved}. 

\subsubsection {Padding process}
A block cipher works on units of a fixed size (known as a block size), but messages come in variety of lengths. So some modes, namely the ECB (Electronic Codebook)  and CBC ones, require that the final block is padded before encryption. In other words, the total number of bits in the plaintext must be a positive multiple of the block size $\mathsf{N}$.

If the data string to be encrypted does not initially satisfy this property, then the formatting of the plaintext must entail an increase in the number of bits. A common way to achieve the necessary increase is to append some extra bits, called padding, to the trailing end of the data string as the last step in the formatting of the plaintext.  An example of a padding method is to append a single $1$ bit to the data string and then to pad the resulting string by as few $0$ bits, possibly none, as are necessary to complete the final block (other methods may be used). 

For the above padding method, the padding bits can be removed unambiguously provided the receiver can determine that the message is indeed padded. One way to ensure that the receiver does not mistakenly remove bits from an unpadded message is to require the sender to pad every message, including messages in which the final block is already complete. For such messages, an entire block of padding is appended.  Alternatively, such messages can be sent without padding if, for each message, the existence of padding can be reliably inferred, \emph{e.g.}, from a message length indicator~\cite{dworkin2001recommendation}.

\subsubsection{CBC mode characteristics}
% Several characteristics  and security properties of the CBC mode are already known. we think it is useful to recall them in order to remind our selves of the important role which was occupied by the the CBC mode in the cryptography field.
Cipher block chaining is a block cipher mode that provides confidentiality but not message integrity in cryptography. 
The CBC mode offers a solution to the greatest part of the problems presented by the ECB (Electronic codebook) for example~\cite{schiltz2003modes}, because thanks to CBC mode the encryption will depends on the context. Indeed, the cipher text of each block encrypted by a keyed encryption function $\varepsilon_k$, where $k$ is the secret key, will depend not only on the initialization vector IV, but also on the plaintext of all preceding blocks.
Specifically, the binary operator XOR is applied between the current bloc of the plaintext and the previous block of the cipher text, as depicted in Figure~\ref{fig:CBC}. Then, we apply the encryption function to the result of this operation. For the first block, the initialization vector takes place of the previous cipher text block.

\begin{figure}[!h]
    \centering
 \subfigure[CBC encryption mode]{\label{fig:CBCenc}
        \includegraphics[scale=0.5]{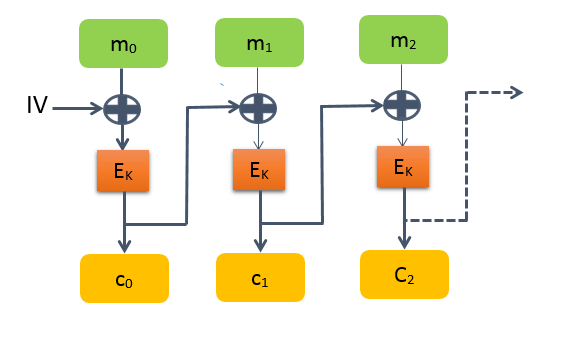}}     \subfigure[CBC decryption mode]{\includegraphics[scale=0.5]{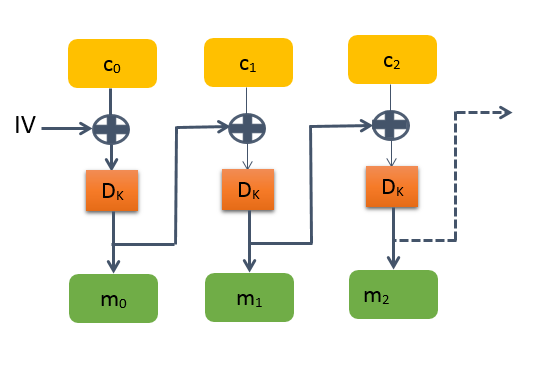}}
    \caption{CBC mode of operation}
     \label{fig:CBC}
\end{figure}

CBC mode has several advantages. In fact, this mode encrypts the same plaintext differently with different initialization vectors. In addition, the encryption of each block depends on the preceding block and therefore, if the order of the cipher text blocks is modified, the decryption will be impossible and the recipient realizes the problem.
Furthermore, if a transmission error affects the encrypted block $C_{i}$, then only the blocks $m_{i}$ and $m_{i+1}$ are assigned, the other blocks will be determined correctly.

CBC has been the most commonly used mode of operation. Its main drawbacks are that encryption is sequential (\emph{i.e.}, it cannot be parallelized), and that the message must be padded to a multiple of the cipher block size. One way to handle this last issue is through the method known as cipher text stealing. Note that a one-bit change in a plaintext or IV affects all following cipher text blocks.

Decrypting with the incorrect initialization vector causes the first block of plaintext to be corrupted, but subsequent plaintext blocks will be correct. This is because a plaintext block can be recovered from two adjacent blocks of cipher text. As a consequence, decryption can be parallelized. Note that a one-bit change on the cipher text causes complete corruption of the corresponding block of plaintext, and inverts the corresponding bit in the following block of plaintext, but the rest of the blocks remain intact. 

%The CBC mode do not shows any detection of integrity.Consequently, it remains vulnerable to active attacks. However, it has a good security against passive attacks. As soon as the encrypted contains two equal block, we ca so deduce information about the original message: if $c_{i}$=$c_{j}$
\section{Proof of chaos}
\label{sec:proof}
%In Section~\ref{sec:proof}, 
In this section, the proof of chaos of the CBC mode of operation is detailed. Modeling follows a same canvas than what has been done for hash functions~\cite{bg10:ij,gb11:bc} or pseudorandom number generation~\cite{bfgw11:ij}.

Let us consider the CBC mode of operation with a keyed encryption function $\varepsilon_k:\mathds{B}^\mathsf{N} \rightarrow \mathds{B}^\mathsf{N} $ depending on a secret key $k$, $\mathsf{N}$ is the size for the block cipher, and $\mathcal{D}_k:\mathds{B}^\mathsf{N} \rightarrow \mathds{B}^\mathsf{N} $ is the associated  decryption function, such that for any k, $\varepsilon_k \circ \mathcal{D}_k$ is the identity function. We define  
the Cartesian product $\mathcal{X}=\mathds{B}^\mathsf{N}\times\mathcal{S}_\mathsf{N}$, where:
\begin{itemize}
\item $\mathds{B} = \{0,1\}$ is the set of Boolean values,
\item $\mathcal{S}_\mathsf{N} = \llbracket 0, 2^\mathsf{N}-1\rrbracket^\mathds{N}$, the set of infinite sequences of natural integers bounded by $2^\mathsf{N}-1$, or the set of infinite $\mathsf{N}$-bits block messages,
\end{itemize}
in such a way that $\mathcal{X}$ is constituted by couples of internal states of the mode of operation together with sequences of block messages.
Let us consider the initial function:
$$\begin{array}{cccc}
 i:& \mathcal{S}_\mathsf{N} & \longrightarrow & \llbracket 0, 2^\mathsf{N}-1 \rrbracket \\
 & (m^i)_{i \in \mathds{N}} & \longmapsto & m^0
\end{array}$$
that returns the first block of a (infinite) message, and the shift function:
$$\begin{array}{cccc}
 \sigma:& \mathcal{S}_\mathsf{N} & \longrightarrow & \mathcal{S}_\mathsf{N} \\
 & (m^0, m^1, m^2, ...) & \longmapsto & (m^1, m^2, m^3, ...)
\end{array}$$
which removes the first block of a message. Let $m_j$ be the $j$-th bit of integer, or block message, $m\in \llbracket 0, 2^\mathsf{N}-1 \rrbracket$ expressed in the binary numeral system, and when counting from the left. We define:
$$\begin{array}{cccc}
F_f:& \mathds{B}^\mathsf{N}\times \llbracket 0, 2^\mathsf{N}-1 \rrbracket & \longrightarrow & \mathds{B}^\mathsf{N}\\
 & (x,m) & \longmapsto & \left(x_j m_j + f(x)_j \overline{m_j}\right)_{j=1..\mathsf{N}} 
\end{array}$$
This function returns the inputted binary vector $x$, whose $m_j$-th components $x_{m_j}$ have been replaced by $f(x)_{m_j}$, for all $j=1..\mathsf{N}$ such that $m_j=0$. In case where $f$ is the vectorial negation, this function will correspond to one XOR between the clair text and the previous encrypted state.  
So the CBC mode of operation can be rewritten as the following dynamical system:
\begin{equation}
\label{eq:sysdyn}
\left\{
\begin{array}{ll}
X^0 = & (IV,m)\\
X^{n+1} = & \left(\mathcal{E}_k \circ F_{f_0} \left( i(X_1^n), X_2^n\right), \sigma (X_1^n)\right)
\end{array}
\right.
\end{equation}
For any given $g:\llbracket 0, 2^\mathsf{N}-1\rrbracket \times \mathds{B}^\mathsf{N} \longrightarrow \mathds{B}^\mathsf{N}$, we denote $G_g(X) = \left(g(i(X_1),X_2);\sigma (X_1)\right)$ (when $g = \mathcal{E}_k\circ F_{f_0}$, we obtain one cypher block of the CBC, as depicted in Figure~\ref{fig:CBC}). So the recurrent relation of Eq.\eqref{eq:sysdyn} can be rewritten in a condensed way, as follows.
\begin{equation}
X^{n+1} = G_{\mathcal{E}_k\circ F_{f_0}} \left(X^n\right) .
\end{equation}
With such a rewriting, one iterate of the discrete dynamical system above corresponds exactly to one cypher block in the CBC mode of operation. Note that the second component of this system is a subshift of finite type, which is related to the symbolic dynamical systems known for their relation with chaos~\cite{Lind:1995:ISD:230198}.

We now define a distance on $\mathcal{X}$ as follows: $d((x,m);(\check{x},\check{m})) = d_e(x,\check{x})+d_m(m,\check{m})$, where:
$$\left\{\begin{array}{ll}
d_e(x,\check{x})  & = \sum_{k=1}^\mathsf{N} \delta (x_k,\check{x}_k)  \\
&\\
d_m(m,\check{m}) & = \displaystyle{\dfrac{9}{\mathsf{N}} \sum_{k=1}^\infty \dfrac{\sum_{i=1}^\mathsf{N} \left|m_i - \check{m}_i\right|}{10^k}} .
\end{array}\right.$$
This distance has been introduced to satisfy the following requirements:
\begin{itemize}
\item The integral part between two points $X,Y$ of the phase space $\mathcal{X}$ corresponds to the number of binary components that are different between the two internal states $X_1$ and $Y_1$.
\item The $k$-th digit in the decimal part of the distance between $X$ and $Y$ is equal to 0 if and only if the $k$-th blocks of messages $X_2$ and $Y_2$ are equal. This desire is at the origin of the normalization factor $\dfrac{9}{\mathsf{N}}$.
\end{itemize}
We will now prove that,
\begin{proposition}
$G_g$ is a continuous map on $(\mathcal{X},d)$.
\end{proposition}

\begin{proof}
Let us consider a sequence $X^n$ of elements of $\mathcal{X}$, which converges to $\check{X}$. We denote by $X^n = (x^n, m^n)$ and $\check{X}=(\check{x},\check{m})$, so that (1) $x^n \longrightarrow \check{x}$ and (2) $m^n \longrightarrow \check{m}$.

Due to (1) and the definition of $d_e$, we have that $\exists n_0 \in \mathds{N}, \forall n \geqslant n_0, x^n=\check{x}$, while (2) and definition of $d_m$ imply that $\exists n_1 \in \mathds{N}, \forall n \geqslant n_1, m_0^n=\check{m}_0$.
Let $\varepsilon > 0$.
\begin{itemize}
\item If $\varepsilon \geqslant 1$, then for $n \geqslant \max (n_0, n_1)$, 

\begin{tabular}{ll}
$d\left(G_g(X^n),G_g(\check{X})\right)$ &  $= d\left( \left( g(i(X_0^n),X_1^n ); \sigma(X_0^n) \right) ;  \left( g(i(\check{X}_0^n),\check{X}_1^n );\sigma(\check{X}_0^n)\right) \right)$ \\
 & 
 $= d\left( \left( g(i(m^n),x^n );\sigma(m^n) \right) ;  \left(  g(i(\check{m}^n),\check{x}^n );\sigma(\check{m}^n)\right) \right)$ \\
  & 
 $= d\left( \left( g(m_0^n,x^n );\sigma(m^n) \right) ;  \left(  g(\check{m}_0^n,\check{x}^n ); \sigma(\check{m}^n)\right) \right)$ \\
   & 
 $= d\left( \left( g(\check{m}_0^n,\check{x}^n );\sigma(m^n) \right) ;  \left(  g(\check{m}_0^n,\check{x}^n ); \sigma(\check{m}^n)\right) \right) $ \\
 &
 $= d_e\left(g(\check{m}_0,\check{x});g(\check{m}_0,\check{x})\right)+d_m(\sigma(m^n),\sigma(\check{m}))$\\
  &
 $= d_m(\sigma(m^n),\sigma(\check{m}))$\\
 & $<1<\varepsilon$.
\end{tabular}
\item Else, there exists an integer $k$ such that $10^{-k}\geqslant \varepsilon > 10^{-(k+1)}$. $m^n \longrightarrow \check{m}$, so $\exists n_2 \in \mathds{N}, \forall n \geqslant n_2, d_m(m^n,\check{m})<10^{-(k+2)}$, \emph{i.e.}, for all terms larger than $n_2$, the $k+2$-th first terms of $m^n$ and $\check{m}$ are equal. So the $k+1$-th first terms of $\sigma(m^n)$ and $\sigma(\check{m})$ are equal too, and thus $d_m\left(\sigma(m^n),\sigma(\check{m}^n)\right) < 10^{-(k+1)}<\varepsilon$.
\end{itemize}
We thus have proven, using the sequential characterization of the continuity, that $G_g$ is continuous on $(\mathcal{X},d)$.
\end{proof}

Let us now recall that a directed graph is strongly connected when it contains a directed path from $u$ to $v$ and a directed path from $v$ to $u$ for every pair of vertices $u$, $v$. Then we have the following proposition.

\begin{proposition}
\label{prop:transitivity}
Let $g=\varepsilon_k \circ F_{f_0}$, where $\varepsilon_k$ is a given keyed block cipher and $f_0:\mathds{B}^\mathsf{N} \longrightarrow \mathds{B}^\mathsf{N}$, $(x_1,...,x_\mathsf{N}) \longmapsto (\overline{x_1},...,\overline{x_\mathsf{N}})$ is the Boolean vectorial negation.
We consider the directed graph $\mathcal{G}_g$, where:
\begin{itemize}
\item vertices are all the $\mathsf{N}$-bit words.
\item there is an edge $m \in \llbracket 0, 2^{\mathsf{N}}-1 \rrbracket$ from $x$ to $\check{x}$ if and only if $g(m,x)=\check{x}$.
\end{itemize}
So if $\mathcal{G}_g$ is strongly connected, then $G_g$ is strongly transitive.
\end{proposition}

\begin{proof}
Let $(x,m)$ and $(\check{x},\check{m})$ be two elements of $\mathcal{X}$, and $\varepsilon >0$. We are looking for $n\in\mathds{N}$ and another point $(x',m')$ close at $\varepsilon$ from $(x,m)$, such that $G_g^n((x',m')) = (\check{x}, \check{m})$.

As $\varepsilon$ may be $<1$, we have necessarily $x'=x$. Let us denote by $n_0=\lfloor log_{10}(\varepsilon) \rfloor +1$, and $\forall j \leqslant n_0,$ $m_j'=m_j$, in such a way that all the points $(x',m')$ having this form are $\varepsilon$-close to $(x,m)$. 
Let $X=G_g^{n_0}(x,m)$.

$\mathcal{G}_g$ being strongly connected, these exist $n_1 \in \mathds{N}$ and edges $M_0, ..., M_{n_1} \in \llbracket 0, 2^\mathsf{N}+1 \rrbracket$ that realize a path between node $X_1$ and $\check{x}$. But the point $$X'= (x; (m_0, m_1,..., m_{n_0}, M_0, M_1, ..., M_{n_1}, \check{m}_0, \check{m}_1,...)$$ is such that:
\begin{itemize}
\item $X'$ is $\varepsilon$-close to $(x,m)$,
\item $G_g^{n_0+n_1}(X')= (\check{x},\check{m})$,
\end{itemize}
proving by doing so the strong transitivity of $G_g$ on $(\mathcal{X},d)$.
\end{proof}

We will now prove that,
\begin{proposition}
\label{prop:regularity}
If $\mathcal{G}_g$ is strongly connected, then $G_g$ is regular.
\end{proposition}

\begin{proof}
Let $(x,m) \in \mathcal{X}$ and $\varepsilon >0$. We are looking for a point $(x',m')$ that is both periodic and $\varepsilon$-close to $(x,m)$.

Let us define $x'=x$, $n_0 = \lfloor log_{10}(\varepsilon )\rfloor +1 $, and $\forall i \leqslant n_0$, $m_i'=m_i$. Consider $X=G_g^{n_0}(x,m)$. $\mathcal{G}_g$ being strongly connected, there exist edges $M_0, ..., M_{n_1}$ that realize a path between $X_1$ and $x$. Finally, the point defined by: 
$$(x; (m_0, ..., m_{n_0},M_0, ..., M_{n_1}, m_0, ..., m_{n_0},M_0, ..., ))$$ is:
\begin{itemize}
\item periodic,
\item $\varepsilon$-close to $(x,m)$.
\end{itemize}
Having found such a point in the neighborhood of $(x,m)$, we thus can claim that $G_g$ is regular.
\end{proof}

According to Propositions~\ref{prop:transitivity} and~\ref{prop:regularity}, we can conclude that,
\begin{theorem}
\label{thm}
If the directed graph $\mathcal{G}_g$ is strongly connected, then the CBC mode of operation is chaotic according to Devaney.
\end{theorem}

\section{Case studies}
\label{sec:caseStudy}

This section is devoted to illustrations on how the above theorem can be applied in practice. A few old school and insecure block ciphers are studied from the standpoint of their dynamics, and chaos property is stated for most of them. This latter is indeed inherited from the conjunction of 3 fundamental operations who are often related to chaos, namely the shift operation, the bitwise exclusive OR, and the vectorial negation~\cite{guyeux13:bc}.

\subsection{A few transposition ciphers}

Let us give now two examples of block ciphers $\varepsilon$ that have a strongly connected directed graph $\mathcal{G}_{\varepsilon \circ f_0}$, leading to a chaotic behavior for the CBC mode of operation. These two rudimentary examples are taken from so-called transposition cipher methods.

\subsubsection{A trivial but chaotic situation}

Let us consider the trivial example where the inputted block cipher method is the identity. This is the most rudimentary transposition cipher where the cyphertext is equal to the plaintext. In that situation, we will show that the CBC encryption mode has a chaotic behavior due to the shift and XOR operations.

Let $Id : \mathds{B}^\mathsf{N} \longrightarrow \mathds{B}^\mathsf{N}$, $x \longmapsto x$ the identity function, which will act as $\varepsilon_k$ in Figure~\ref{fig:CBCenc}. We consider $g=Id \circ F_{f_0}$ defined from $\mathds{B}^\mathsf{N} \times \llbracket 0, 2^\mathsf{N}-1 \rrbracket$ to $\mathds{B}^\mathsf{N}$, and the graph $\mathcal{G}_g$ as defined in the previous section. We will show that this directed graph is strongly connected, thus leading to a chaotic behavior of the CBC mode of operation.

\begin{lemma}
The directed graph $\mathcal{G}_{Id \circ F_{f_0}}$ is strongly connected.
\end{lemma}

\begin{proof}
\label{proof:id}
Let us consider two nodes $x$ and $\check{x}$. Let $k$ be the number of different bits in the $\mathsf{N}$-bits binary decomposition of these two integers. We denote by $m_1, ..., m_k \in \llbracket 1, \mathsf{N} \rrbracket$ the positions of these differences. So $g(m_1, g(m_2, ..., g(m_k,x)...))= \check{x}$, as when considering $\varepsilon_k = Id$, the effect of $g(m,X)$ is to switch the $m$-th binary digit of $X$. 

In other words, there is a path between $x$ and $\check{x}$, namely the one that follows edges $m_1, m_2,...$, and then $m_k$. Hence, $\mathcal{G}_{Id \circ F_{f_0}}$ is strongly connected.
\end{proof}

According to above lemma and Theorem~\ref{thm}, we can thus conclude that,

\begin{proposition}
The CBC mode of encryption has a chaotic behavior when the embedded block cipher is the identity.
\end{proposition}

This result may appear as surprising, but it only states that, even if we do not operate any cipher operation at each internal state of the system, the xor operation between the previous state and a new block of the message introduces a sufficient amount of disorder to be unable to predict the future evolution of the internal state. To say this differently, a modification of one block of message, or of the IV, may possibly have effects on the outputted block that cannot be predicted on a long term basis.

\subsubsection{A second transposition cypher method}

Let us now consider that the encryption function is the vectorial negation: $\varepsilon_k = f_0$. Remark that, again in that simple case, we do not consider any key $k$. In that situation, the outputs have a chaotic dependence on the input: any slightly modification on either the input vector or on one block message have effects on the output that cannot be predicted. This chaotic property is a direct consequence of the following lemma:

\begin{lemma}
Suppose that $\mathsf{N}$ is even. So the directed graph $\mathcal{G}_{f_0 \circ F_{f_0}}$ is strongly connected.
\end{lemma}

\begin{proof}
Let $g=f_0 \circ F_{f_0}$, so 
$$\begin{array}{ll}
g(m,(x_1, ..., x_\mathsf{N})) & = f_0((x_1, ...,x_{m-1},\overline{x_m},x_{m+1}, ..., x_\mathsf{N}))\\
& = (\overline{x_1}, ...,\overline{x_{m-1}},x_m,\overline{x_{m+1}}, ..., \overline{x_\mathsf{N}}) .
\end{array}$$
In other words, the effect of $g$ on $(m,(x_1, ..., x_\mathsf{N}))$ is to switch each bit in $x = (x_1, ..., x_\mathsf{N})$ (that is, to operate the binary negation), except on the $m$-th one.
In other words, in $\mathcal{G}_{f_0 \circ F_{f_0}}$, there is an edge labelled $m$ between two nodes $x_1$ and $x_2$ if and only if only the $m$-th bit in their binary decomposition are equal (all the other bits must be different).

Obviously, starting from $x=(x_1, ...,x_{m-1}, x_m,x_{m+1}, ..., x_\mathsf{N})$ and following edges 2, 3, ..., $\mathsf{N}$ will switch each bit inside vector $x$ an even number of times, except for the first component, which is switched an odd number of times. So the state that is recovered after such a path is $(\overline{x_1}, ...,x_{m-1},x_m,x_{m+1}, ..., x_\mathsf{N})$. Similarly, when we start from $x$ and we follow edges labeled by 1, 3, 4, ..., $\mathsf{N}$, then we arrive on a state having all the same component of $x$, except the second one.

Given any couple $(x, \check{x})$ of nodes in $\mathcal{G}_{f_0 \circ F_{f_0}}$, we can thus extract the position $m_0, m_1..., m_k$ of the differences in their binary decomposition as in Proof~\ref{proof:id}. Then, the path in this graph that starts by $x$ and then following the edges:
\begin{itemize}
\item 1, 2, ..., $\mathsf{N}$, except $m_0$
\item following by 1, 2, ..., $\mathsf{N}$, except $m_1$
\item ...
\item following by 1, 2, ..., $\mathsf{N}$, except $m_k$
\end{itemize}
arrives in $\check{x}$, proving by doing so the strong connectivity of $\mathcal{G}_{f_0 \circ F_{f_0}}$.
\end{proof}

The chaotic behavior of the CBC mode of operation can be deduced again from the above lemma, in the case where $\varepsilon_k = f_0$.

\subsection{Caesar shift}

\begin{figure}[!h]
    \centering
 \subfigure[$\mathsf{N}=2$, $k=1$]{\label{fig:CBCenc}
        \includegraphics[scale=0.35]{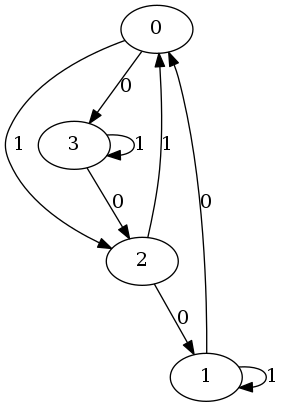}}     \subfigure[$\mathsf{N}=2$, $k=2$]{\includegraphics[scale=0.35]{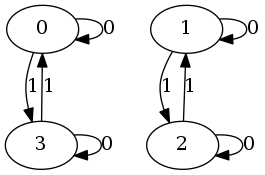}}
    \caption{$\mathcal{G}_g$ of some Caesar ciphers $\mathcal{E}_k(x):\llbracket 0, 2^\mathsf{N}-1 \rrbracket \longrightarrow \llbracket 0, 2^\mathsf{N}-1 \rrbracket , x \longmapsto x+k \mod 2^\mathsf{N}$}
     \label{fig:CBC2}
\end{figure}

We now consider one of the simplest and most widely known substitution cipher, namely the Caesar shift. 
In this latter, each symbol in the plaintext is replaced by a symbol some fixed number of positions down the given alphabet. Translated in the $\mathsf{N}$ binary digits set of integers, this cypher can be written as follows:
\begin{equation}
\begin{array}{cc}
\begin{array}{cccc}
\mathcal{E}_k(x):& \llbracket 0, 2^\mathsf{N}-1 \rrbracket& \longrightarrow &\llbracket 0, 2^\mathsf{N}-1 \rrbracket\\
& x & \longmapsto & x+k \mod 2^\mathsf{N} \end{array}
&;
\begin{array}{cccc}
\mathcal{D}_k(x):& \llbracket 0, 2^\mathsf{N}-1 \rrbracket& \longrightarrow &\llbracket 0, 2^\mathsf{N}-1 \rrbracket\\
& x & \longmapsto & x-k \mod 2^\mathsf{N} \end{array}
\end{array}
\end{equation}
where $k$ is the shift value acting as secret key. We will now show through examples that the CBC mode of operation embedding the Caesar shift can behave either chaotically or not, depending on $k$ and $\mathsf{N}$.

\begin{table}[]
\centering
\begin{tabular}{cc}
\begin{tabular}{c|c|c|c}
$x$ & $m$ & $F_{f_0}(x,m)$ & $g(m,x) = \mathcal{E}_k \circ F_{f_0}(x,m)$ \\
\hline
 0 (0,0) & 0 & 2 (1,0) & 3 \\
 0 (0,0) & 1 & 1 (0,1) & 2 \\
 1 (0,1) & 0 & 3 (1,1) & 0 \\
 1 (0,1) & 1 & 0 (0,0) & 1 \\
 2 (1,0) & 0 & 0 (0,0) & 1 \\
 2 (1,0) & 1 & 3 (1,1) & 0 \\
 3 (1,1) & 0 & 1 (0,1) & 2 \\
 3 (1,1) & 1 & 2 (1,0) & 3 \\
 \end{tabular}
&
\begin{tabular}{c|c|c|c}
$x$ & $m$ &  $F_{f_0}(x,m)$ & $g(m,x)$ \\
\hline
 0 (0,0) & 0 & 2 (1,0) & 0 \\
 0 (0,0) & 1 & 1 (0,1) & 3 \\
 1 (0,1) & 0 & 3 (1,1) & 1 \\
 1 (0,1) & 1 & 0 (0,0) & 2 \\
 2 (1,0) & 0 & 0 (0,0) & 2 \\
 2 (1,0) & 1 & 3 (1,1) & 1 \\
 3 (1,1) & 0 & 1 (0,1) & 3 \\
 3 (1,1) & 1 & 2 (1,0) & 0 \\
 \end{tabular}
\\
$\mathsf{N}=2$, $k=1$ & $\mathsf{N}=2$, $k=2$
\end{tabular}
\caption{$g(x,m)$ for two shifts in Caesar cipher}
\label{tab:Caesar}
\end{table}

Table~\ref{tab:Caesar} contains the $g(x,m)$ values for a shift of 1 and 2 respectively, in Caesar cipher over $2$-bit blocks. In this table, we can see that, when a shift of 2 value is operated in the Caesar based CBC, it is impossible to reach the blocks (1,0) and (1,1) when starting from (0,0). This fact, which is illustrated in Figure~\ref{fig:CBC2}, leads to non chaotic behavior of the CBC mode of encryption, as the graph $\mathcal{G}_g$ is not strongly connected. Conversely, this graph becomes strongly connected when operating a shift of 1, as depicted in the same figure, leading to a chaotic behavior of the mode of encryption.

\begin{figure}[!h]
    \centering
 \subfigure[$\mathsf{N}=3$, $k=1$]{\label{fig:CBCenc}
        \includegraphics[scale=0.35]{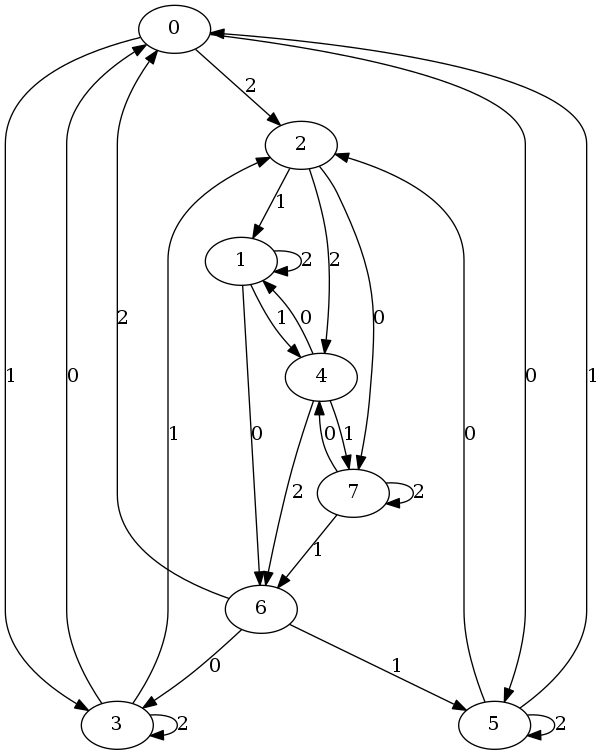}}     \subfigure[$\mathsf{N}=3$, $k=2$]{\includegraphics[scale=0.35]{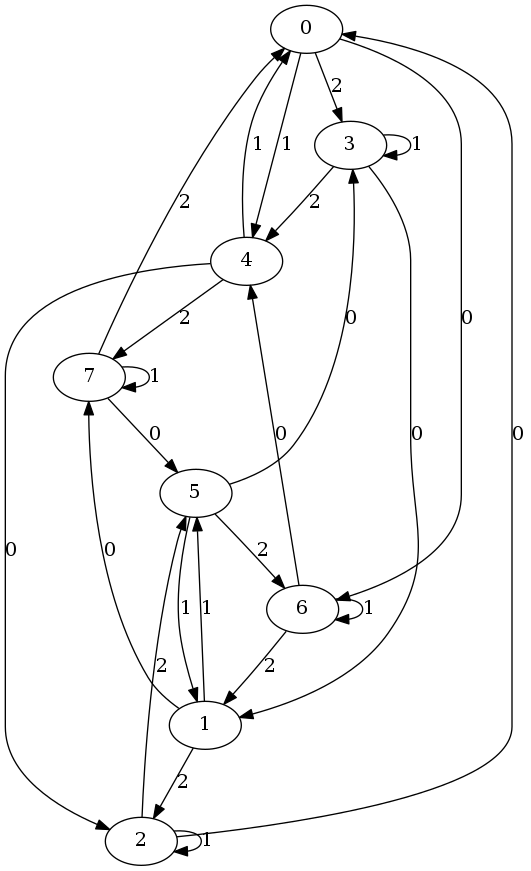}}
    \caption{$\mathcal{G}_g$ of some Caesar ciphers $\mathcal{E}_k(x):\llbracket 0, 2^\mathsf{N}-1 \rrbracket \longrightarrow \llbracket 0, 2^\mathsf{N}-1 \rrbracket , x \longmapsto x+k \mod 2^\mathsf{N}$}
     \label{fig:CBC3}
\end{figure}

Figure~\ref{fig:CBC3}, for its part, presents the graph of iteration of the Caesar based CBC mode of operation, with the same kind of shifts, but when operating on blocks of size 3. 
We can verify that, at each time, the cipher block chaining behaves chaotically. In that situation, we can guarantee that any error on the IV (starting state) or on the message to encrypt (edges to browse) may potentially lead to a completely different list of visited states, that is, of a completely different ciphertext.

\section{Discussion}

It has been proven in the previous section that some 
well chosen block ciphers can lead to a chaotical behavior for the CBC mode of operation. Indeed, this mode of operation can be seen as a discrete dynamical system (recurrent sequence), whose evolution can thus be studied using common tools taken from the mathematical analysis. In particular, its dynamics can be deeply investigated, both qualitatively and quantitatively, using the rigorous mathematical topology field of research. A subfield of this discipline specifically focus on complex dynamics, defining when such recurrent sequences have evolution over time that cannot be easily understood, whose effects of a slight alteration of the initial term of the sequence may potentially lead to unpredictable mid term evolution, and so on.

This subfield, namely the mathematical theory of chaos, has known several developments in various directions these last four decades. Being mature now, this theory allows to measure the complexity of the dynamics by evaluating, for instance, the expansivity, the sensibility, the entropy of the dynamical system, or its ability to visit the whole space, to name a few~\cite{bahi2011efficient}. These properties may impact, for instance, the ability of the mode of operation to face side-channel attacks, or to measure the stability of the mode operation against transmission errors.

Devaney's definition of chaos is one of the oldest formalization of such complex dynamics, and a lot of new notions have been introduced since this first approach of chaos. Similarly, the block ciphers investigated in the previous section are out-of-date. Indeed, the interest of our work is not to provide a collection of secure and complex CBCs, but to initiate a complementary approach for studying such modes of operation. Our intention is to show how to model such modes, and that it is possible to study the complexity of their dynamics. Up-to-date block ciphers and modes of operation, together with topological analyses using most recent developments in this field, need to be investigated, while the interest of each topological property of complexity must be related to desired objectives for each mode of operation.
+
\section{Conclusion and future work}

In this paper, we proved that CBC mode of operation behaves as Devaney's topological chaos if the iteration function used is the vectorial Boolean negation. We applied these results to proof the chaotic behaviour of the CBC mode of operation. The vectorial Boolean negation function has been chosen here, but the process remains general and other iterate functions g can be used. The sole condition is to prove that $G_g$ satisfies the Devaney's chaos property.
 Thereafter, we have given some examples of keyed block ciphers, which can be used by the CBC mode, and which can lead to a chaotic behavior for this mode seeing that they have a strongly connected directed graph. These examples are taken from so-called transposition cipher methods.
 
In future work, we will investigate other choices of keyed block ciphers. This same canvas of our contribution will be explored to proof the chaotic behaviour of the other modes of operation. We will treat the DES (Data encryption standard block) cipher algorithm so as to put in evidence his chaotic also behavior. Doing theses contributions will allow us to enrich the block cipher algorithm's field.

%\begin{acknowledgements}
%\end{acknowledgements}



\bibliographystyle{plain}
\bibliography{biblio}
\end{document}